\documentclass[11pt]{article}

\newcommand{\version}{long}

\usepackage{etoolbox,ifthen,tabulary}

\newcommand{\ifft}{\ifundef{\shortiff}{if and only if }{iff }}
\newcommand{\spc}{{ }}

\newcommand{\indraft}[1]{\ifthenelse{\equal{\version}{draft}}{#1}{}}
\newcommand{\infinal}[1]{\ifthenelse{\equal{\version}{final}}{#1}{Only shown in the final version}}
\newcommand{\inshort}[1]{\ifthenelse{\equal{\version}{short}}{#1}{}}
\newcommand{\inlong}[1]{\ifthenelse{\equal{\version}{long}}{#1}{}}

\newcommand{\instacsornot}[2]{\ifdef{\stacs}{#1}{#2}}

\newcommand{\notinstacs}[1]{\instacsornot{}{#1}}

\newcommand{\appref}[1]{Appendix~\ref{app:#1}}
\newcommand{\secref}[1]{Section~\ref{sec:#1}}
\newcommand{\ssecref}[1]{Subsection~\ref{ssec:#1}}

\newcommand{\renameenv}[2]{
  \expandafter\let\csname #1#2\expandafter\endcsname
  \csname #1\endcsname
  \expandafter\let\csname end#1#2\expandafter\endcsname
  \csname end#1\endcsname
  \expandafter\let\csname #2\endcsname\relax
  \expandafter\let\csname end#2\endcsname\relax}

\ifundef{\defaultlists}{
  \usepackage[inline,shortlabels]{enumitem}
  \setenumerate[1]{(a),itemsep=0pt,topsep=3pt,parsep=0pt,partopsep=0pt}
  \setenumerate[2]{(i),noitemsep,topsep=3pt,parsep=0pt,partopsep=0pt}
  \setenumerate[3]{(A),noitemsep,topsep=3pt,parsep=0pt,partopsep=0pt}
  \setenumerate[4]{(I),noitemsep,topsep=3pt,parsep=0pt,partopsep=0pt}
  \setitemize{noitemsep,topsep=3pt,parsep=0pt,partopsep=0pt}
  \setdescription{noitemsep,topsep=3pt,parsep=0pt,partopsep=0pt}
  \setlist{noitemsep,topsep=3pt,parsep=0pt,partopsep=0pt}}{}

\newcolumntype{x}[1]{>{\centering\arraybackslash}m{#1}}
\usepackage{fixltx2e} % Fixes "ERROR: LaTeX Error: Unknown float option `H'."
\usepackage{algorithm,algpseudocode,amsmath,amssymb,amsthm,complexity,etoolbox,graphicx,hyperref,ifthen,mathtools,url,thmtools,thm-restate}
\notinstacs{\usepackage{fullpage}}
\usepackage[inline,shortlabels]{enumitem}

\usepackage{complexity}

\newclass{\ioPSPACE}{i.o.\text{-}PSPACE}
\newlang{\Halt}{Halt}
\usepackage{algpseudocode,amsmath,amssymb,etoolbox,fixltx2e,mathtools}

\let\eqref\relax

\DeclareMathOperator*{\im}{Im}

\DeclareFontFamily{U}{mathx}{\hyphenchar\font45}
\DeclareFontShape{U}{mathx}{m}{n}{
      <5> <6> <7> <8> <9> <10>
      <10.95> <12> <14.4> <17.28> <20.74> <24.88>
      mathx10
      }{}
\DeclareSymbolFont{mathx}{U}{mathx}{m}{n}
\DeclareMathSymbol{\bigtimes}{1}{mathx}{"91}

\algnewcommand{\Input}{\item[\textbf{Input:}]}
\algnewcommand{\Output}{\item[\textbf{Output:}]}
\MakeRobust{\Call} % Fix a problem with nested calls.

\newcommand\restr[2]{{
  \left.\kern-\nulldelimiterspace
  #1
  \vphantom{\big|}
  \right|_{#2}
  }}

\newcommand{\abs}[1]{\left\vert#1\right\vert}

\newcommand{\gl}{\mathrm{GL}}
\newcommand{\spl}{\mathrm{SL}}
\newcommand{\diag}{\mathrm{diag}}

\newcommand{\aut}{\mathrm{Aut}}
\newcommand{\hol}{\mathrm{Hol}}
\renewcommand{\hom}{\mathrm{Hom}}
\newcommand{\edm}{\mathrm{End}}

\newcommand{\sym}{\mathrm{Sym}}

\newcommand{\iso}{\mathrm{Iso}}

\newcommand{\rad}{\mathrm{Rad}}

\newcommand{\setb}[2]{\left\{#1 \;\middle|\; #2\right\}}

\newcommand{\nth}[1]{\ensuremath{{#1}^{\mathrm{th}}}}
\newcommand{\thmref}[1]{Theorem~\ref{thm:#1}}
\newcommand{\lemref}[1]{Lemma~\ref{lem:#1}}

\newcommand{\defref}[1]{Definition~\ref{defn:#1}}

\newcommand{\propref}[1]{Proposition~\ref{prop:#1}}

\newcommand{\eqref}[1]{(\ref{eq:#1})}

\newcommand{\bma}{\mathbf{a}}
\newcommand{\bmb}{\mathbf{b}}

\newcommand{\bmg}{\mathbf{g}}
\newcommand{\bmh}{\mathbf{h}}

\newcommand{\bmx}{\mathbf{x}}
\newcommand{\bmy}{\mathbf{y}}

\newcommand{\bbF}{\mathbb{F}}

\newcommand{\bbZ}{\mathbb{Z}}

\newcommand{\cB}{\mathcal{B}}
\newcommand{\cC}{\mathcal{C}}

\newcommand{\ra}{\rightarrow}

\newcommand{\tril}{\triangleleft}

\inshort{\usepackage{amsthm}
\makeatletter
\def\thm@space@setup{\thm@preskip=3pt \thm@postskip=3pt}
\makeatother

\makeatletter
\renewenvironment{proof}[1][\proofname]{\par
  \pushQED{\qed}
  \normalfont
  \topsep3pt \partopsep0pt % No space before
  \trivlist
  \item[\hskip\labelsep
        \itshape
    #1\@addpunct{.}]\ignorespaces
  }{
    \popQED\endtrivlist\@endpefalse
    \addvspace{0pt plus 0pt} % No space after
  }
\makeatother
} %$
\usepackage{etoolbox,ifthen,thmtools,thm-restate}

\ifundef{\dontnumberwithin}{\declaretheorem[numberwithin=section]{dummy}}{\declaretheorem{dummy}} % Using dummy resolves some incompatibility issues with certain documentclasses.
\declaretheorem[sibling=dummy]{theorem}

\declaretheorem[sibling=dummy]{lemma}

\declaretheorem[sibling=dummy]{definition}

\declaretheorem[sibling=dummy]{proposition}

\ifundef{\defaultthmcontinues}{\renewcommand{\thmcontinues}[1]{}}{}
\inshort{\newcommand{\shortiff}{}}
\inshort{\usepackage[compact]{titlesec}}

\newclass{\BI}{BI}
\newclass{\CI}{CI}
\newclass{\CIS}{CI^*}
\newclass{\GIS}{GI^*}
\newclass{\GRI}{GrI}
\newclass{\GRIS}{GrI^*}
\newclass{\NGRI}{NGrI}
\newclass{\NGRIS}{NGrI^*}

\newcommand{\cart}{\times_c}

\newcommand{\bone}{\mathbf{1}}

\begin{document}

\title{On the Group and Color Isomorphism Problems}
\author{Fran\c{c}ois Le Gall \\
  {\small Graduate School of Informatics} \\
  {\small Kyoto University} \\
  {\small Email: legall@i.kyoto-u.ac.jp} \\
  \and
  David J. Rosenbaum \\
  {\small The University of Tokyo} \\
  {\small Department of Computer Science} \\
  {\small Email: djr7c4@gmail.com}}
\date{September 27, 2016}

\maketitle
\thispagestyle{empty}

\begin{abstract}
  In this paper, we prove results on the relationship between the complexity of the group and color isomorphism problems.  The difficulty of color isomorphism problems is known to be closely linked to the the composition factors of the permutation group involved.  Previous works are primarily concerned with applying color isomorphism to bounded degree graph isomorphism, and have therefore focused on the alternating composition factors, since those are the bottleneck in the case of graph isomorphism.

  We consider the color isomorphism problem with composition factors restricted to those other than the alternating group, show that group isomorphism reduces in $n^{O(\log \log n)}$ time to this problem, and, conversely, that a special case of this color isomorphism problem reduces to a slight generalization of group isomorphism.  We then sharpen our results by identifying the projective special linear group as the main obstacle to faster algorithms for group isomorphism and prove that the aforementioned reduction from group isomorphism to color isomorphism in fact produces only cyclic and projective special linear factors.  Our results demonstrate that, just as the alternating group was a barrier to faster algorithms for graph isomorphism for three decades, the projective special linear group is an obstacle to faster algorithms for group isomorphism.
\end{abstract}

\inlong{
  \newpage
  \setcounter{page}{1}}

\section{Introduction}
\label{sec:intro}
The complexity of isomorphism testing problems is worthy of study both because they are fundamental computational questions and also because many of them are not known to be in \P, but nevertheless appear to be easier than the \NP-complete problems.  The most heavily studied of these is the \emph{graph isomorphism problem}.  It is strongly suspected that graph isomorphism is not \NP-complete both because this would imply the collapse of the polynomial hierarchy~\cite{babai1985a,boppana1987a,goldwasser1989a,goldreich1991a} and also because there are subexponential time algorithms~\cite{luks1982a,babai1983a,babai1983b,babai2016a} for testing isomorphism of general graphs, which is much better than the $2^{O(n^2)}$ time complexity that we would expect based on the exponential-time hypothesis~\cite{impagliazzo1999a}.

For more than three decades, the $2^{O(\sqrt{n \log n})}$ time bound from~\cite{luks1982a,babai1983a,babai1983b} was the best known for testing isomorphism of general graphs in the worst case.  This result is based on
\begin{enumerate*}
\item reducing testing isomorphism of a pair of arbitrary graphs to testing isomorphism of many pairs of graphs of degree $\sqrt{n / \log n}$ using Zemlyachenko's lemma (cf.~\cite{babai1981a}), and
\item reducing~\cite{luks1982a,babai1983a,babai1983b} testing isomorphism of graphs of degree at most $d$
\end{enumerate*}
to another problem known as the \emph{color automorphism problem}~\cite{luks1982a,babai1983a,babai1983b}.

In this problem, we are given a set $X$ of size $n$ and a coset $\sigma \Gamma$ where $\Gamma$ is a subgroup of the symmetric group on $X$ and $\sigma$ is a permutation of $X$.  We are also given a function $f : X \ra [n]$ that specifies the color of each element of $X$.  The problem is to compute all $\pi \in \sigma \Gamma$ such that $f (\pi x) = f(x)$ for all $x \in X$.  It is important to note that the set of all such $\pi$ forms a subcoset of $\sigma \Gamma$~\cite{luks1982a}, so the solution can be represented compactly as a coset representative along with the generators of a subgroup of $\Gamma$.  In this paper, we consider a slight generalization of the color automorphism problem in which there are two functions $f_1$ and $f_2$.

\begin{definition}
  \label{defn:color-iso}
  In the \emph{color isomorphism problem}, we are given a set $X$ of size $n$ and a coset $\sigma \Gamma$ containing permutations of $X$ and two functions $f_1 : X \ra [n]$ and $f_2 : X \ra [n]$.  The goal is to find all $\pi \in \sigma \Gamma$ such that $f_2 (\pi x) = f_1(x)$ for all $x \in X$.
\end{definition}

Note that we can recover the color automorphism problem by stipulating that $f_1 = f_2$.
%We are also functions $f_1 : X \ra [n]$ and $f_2 : Y \ra [n]$ that encode the colors.  The problem is to compute all $\pi \in \sigma \Gamma$ such that $f_2 (\pi x) = f_1(x)$ for all $x \in X$.  It is important to note that the set of all such $\pi$ forms a subcoset~\cite{luks1982a} of $\sigma \Gamma$, so the solution can be represented compactly as a coset representative along with the generators of a subgroup of $\Gamma$.
%In this problem, we are given two sets $X$ and $Y$ of size $n$ and a coset $\sigma \Gamma$ where $\Gamma$ is a subgroup of $\sym(X \uplus Y)$ such that every element of $\sigma \Gamma$ maps $X$ to $Y$.  We are also functions $f_1 : X \ra [n]$ and $f_2 : Y \ra [n]$ that encode the colors.  The problem is to compute all $\pi \in \sigma \Gamma$ such that $f_2 (\pi x) = f_1(x)$ for all $x \in X$.  It is important to note that the set of all such $\pi$ forms a subcoset~\cite{luks1982a} of $\sigma \Gamma$, so the solution can be represented compactly as a coset representative along with the generators of a subgroup of $\Gamma$.

The complexity of the color isomorphism problem is strongly dependent on the non-Abelian composition factors (which we define later) of the group $\Gamma$: if every non-Abelian composition factor of $\Gamma$ is isomorphic to a subgroup of the symmetric group on $d$ elements, then any color isomorphism problem on a subcoset $\sigma \Gamma$ can be solved in $n^{O(d / \log d)}$ time~\cite{babai1983a,babai1983b}\footnote{In these previous papers, color automorphism and string canonization were the problems that were considered.  However, the techniques in these papers can easily be adapted to solve color isomorphism within the same bounds.}.  The barrier to improving the $n^{O(d / \log d)}$ time bound, and hence also the $2^{O(\sqrt{n \log n})}$ time bound for graph isomorphism, depends only on the composition factors that are isomorphic to alternating groups~\cite{babai1983b}.  All other composition factors can be handled simply by brute force in $n^{O(\log n)}$ time by a result of Pyber~\cite{pyber1991a}.

% groundbreaking result that makes the first significant progress on worst-case graph isomorphism testing in three decades
In a recent paper~\cite{babai2016a}, Babai overcame the obstacle of the alternating composition factors with an algorithm that solves the graph and color isomorphism problems in $2^{O(\log^{2 + O(1)} n)}$ time.  This is almost an exponential speedup.  However, the statement of this result does not allow us to obtain speedups for other types of composition factors since --- as mentioned above --- they can be dealt with exhaustively in $n^{O(\log n)}$ time.

In this work, we study the complexity of color isomorphism problems involving composition factors other than the alternating group.  We accomplish this by comparing this class of color isomorphism problems to the \emph{group isomorphism problem} --- a fundamental problem in computational group theory that has been well studied and has seen a surge of activity in the last few years~\cite{lipton1977a,savage1980a,vikas1996a,arvind2003a,kavitha2007a,legall2008a,qiao2011a,babai2011a,codenotti2011a,babai2012a,babai2012b,lewis2012a,grochow2013a,papakonstantinou2014a,grochow2015a}.  As we will show, group isomorphism depends only on the cyclic and projective special linear composition factors, so we focus our attention on color isomorphism problems with cyclic and projective special linear composition factors.  Our first result shows that, just as the alternating group was a barrier to placing graph isomorphism in quasi-polynomial time for more than thirty years, the projective special linear group is a barrier to faster algorithms for group isomorphism.  Before we can present our reduction, we need to introduce some notation.  Let \CI\spc denote the class of all color isomorphism problems and let \CIS\spc be all color isomorphism problems with cyclic and projective special linear composition factors.  We denote the group isomorphism problem by \GRI.

\begin{restatable}{theorem}{redgroup}
  \label{thm:red-group}
  \GRI\spc is Turing reducible to \CIS\spc in $n^{O(\log \log n)}$ time.
\end{restatable}

It is important to note that this result is trivial when \CIS\spc is replaced with \CI.  The main difficulty here is to restrict the non-Abelian composition factors to the projective special linear group.

Our proof is based a holomorph trick suggested by Babai (personal communication), a new notation for dealing with iterated wreath products, structural results on automorphisms of finite\footnote{All of the groups that we deal with in this paper are finite and we shall omit the adjective finite from now on.} Abelian groups (which we prove using~\cite{hillar2006a}) and the algorithm of~\cite{babai2011a}.  The holomorph trick can also be replaced by the framework introduced by Luks' in his recent paper~\cite{luks2015a} that shows how to test isomorphism of composition series in polynomial time.  We discuss the relationship between our work and Babai's holomorph trick in more detail in \ssecref{prev-comp}.

Our next result is the simple (but to our knowledge previously unknown) observation\inlong{\footnote{The proof of this result is straightforward.  However, for the sake of completeness, we give a proof in \appref{color-graph}.}} that the color isomorphism problem is equivalent to a slight generalization of graph isomorphism.  The main reasons for mentioning this result are to make the relationship between general graph isomorphism and color isomorphism explicit and also to motivate one of our later results in this paper.  Let \GIS\spc be the problem of computing all isomorphisms between two graphs $X$ and $Y$ that are contained in a specified subcoset $\sigma \Gamma$ that maps the vertices of $X$ to the vertices of $Y$.

\begin{restatable}{theorem}{cigis}
  \label{thm:cigis}
  \CI\spc and \GIS\spc are equivalent under polynomial-time many-one reductions.
\end{restatable}

Next, we explore the question of how much more difficult \CIS\spc is compared to \GRI.  To do this, we introduce a slight generalization of group isomorphism and show that a special case of \CIS\spc can be reduced to it.  To this end, we define \GRIS\spc to be the problem of computing all isomorphisms from a group $G$ to a group $H$ that are contained within a coset $\sigma \Gamma$ that maps the elements of $\Gamma$ to $H$ analogously to \GIS\spc for graphs.  It seems unlikely that \GRIS\spc is much harder than \GRI; currently, the fastest worst-case algorithms for both problems run in $n^{O(\log n)}$ time.  Before we can define the special case \BI\spc of \CIS\spc that we will reduce to \GRIS, we need to introduce some additional terminology.

\begin{definition}
  Let $f : B \times B \ra A$ be a bilinear map defined on Abelian groups (regarded as $\bbZ$-modules).  Then an \emph{isometry} is a map $\alpha \in \aut(B)$ such that $f(x, y) = f(\beta x, \beta y)$ for all $x, y \in B$.
\end{definition}

It is easy to see that the isometries of a bilinear map form a group.  Now we can define the problem \BI.

\begin{definition}
  \label{defn:bi}
  Let $A$ and $B$ be Abelian groups and let $f : B \times B \ra A$ be a bilinear map given as a table of the values $f(x, y)$ for all $x, y \in B$.  Then \BI\spc is the problem of computing the isometry group of $f$.
\end{definition}

This redundancy in the representation of $f$ is similar in spirit to the Cayley table representation used in group isomorphism.  Another reason to use this redundant representation is that it means that \BI\spc corresponds to a color isomorphism problem involving an action of $\aut(B)$ on the Abelian group $B \times B \times A$.  Since the composition factors of $\aut(B)$ are either cyclic or projective special linear, this implies that \BI\spc is a special case of \CIS.

Versions of this problem in which the bilinear map is specified compactly as a matrix and $A$ and $B$ are vector spaces have been studied.  Brooksbank and Wilson showed~\cite{brooksbank2012a} that for bilinear maps that are Hermitian, one can compute the isometry group in polynomial time.  Hermitian matrices generalize the symmetric and skew-symmetric matrices, but there are many matrices that are not hermitian.

Our next result shows that \BI\spc is polynomial-time many-one reducible to \GRIS.

\begin{restatable}{theorem}{bigris}
  \label{thm:bi-gris}
  \BI\spc is polynomial-time many-one reducible to \GRIS.
\end{restatable}

Previously, Grochow and Qiao studied~\cite{grochow2013a,grochow2015a} a generalization of the reverse direction of this reduction and used it to prove several interesting results.  Our result complements theirs, and, to our knowledge, is the first reduction to (a slight generalization of) group isomorphism.

Our proof is based on using cohomology to construct a group $G_f$ from the bilinear map $f$ that contains $A$ as a normal subgroup. We then compute a certain subgroup of the automorphism group of $G_f$ by solving a problem in \GRIS.  We show that every automorphism $\phi$ in this subgroup defines a map $\varphi_{\phi} : G_f / A \ra A$ and prove that $\phi$ gives rise to an isometry of $f$ precisely when $\varphi_{\phi}$ is a homomorphism.  By further restricting the subcoset in the instance of \GRIS\spc above, we can ensure that every element of the resulting subgroup of automorphisms gives rise to a homomorphism and that every isometry can be obtained from an automorphism in this subgroup.

Additionally, \BI\spc is likely to be equivalent to \GRI.  The reason is as follows.  The hard case of \GRI\spc is conjectured to be testing isomorphism of nilpotent groups of class $2$ (\NGRI).  \thmref{bi-gris} in fact reduces to \NGRIS, where \NGRIS\spc is defined analogously to \GRIS.  Moreover, one can show that \NGRI\spc reduces to the problem of computing all $\alpha \in \aut(A)$, $\beta \in \aut(B)$ and $b \in \cB$ such that $f(x, y) = \alpha f(\beta x, \beta y) + b$ for all $x, y \in B$ where $A$, $B$ are abelian groups, $\cB$ is a known subgroup of the group of bilinear maps from $B \times B$ to $A$ and $f : B \times B \ra A$ is bilinear.  One can then recover \BI\spc as a special case by setting $\alpha = 1$ and $b = 0$.  This only removes cyclic composition factors from the resulting corresponding color isomorphism problem.

\section{Background}
\label{sec:background}
In this section, we introduce some of the basic group theoretic concepts used later in the paper.  We also discuss related results on testing isomorphism of composition series.

\subsection{Group theory background}
%A subgroup $H$ of a group $G$ is a normal subgroup (denoted $H \tril G$) if $g h g^{-1} \in H$ for all $g \in G$ and $h \in H$.
A \emph{subnormal} series of a group $G$ is a chain of subgroups $G_1 = 1 \tril G_2 \tril \cdots \tril G_k = G$ where each subgroup is normal in the next and $1$ denotes the trivial subgroup.  The \emph{factor groups} of this series are the groups $G_{i + 1} / G_i$.  For a group $G$, let $[G, G]$ be the subgroup of $G$ generated by the \emph{commutators} $[g_1, g_2] = g_1 g_2 g_1^{-1} g_2^{-1}$ where $g_1, g_2 \in G$.  One series that will be of interest is the \emph{derived series} $G^{(k)} \tril \cdots \tril G^{(0)} = G$.  Here, $G^{(0)} = G$ and each $G^{(i + 1)} = [G^{(i)}, G^{(i)}]$ and $k$ is the smallest natural number such that $G^{(k + 1)} = G^{(k)}$.  It need not be the case that $G^{(k)} = 1$.  If this holds, then $G$ is a \emph{solvable group}.

If a subnormal series is maximal so that no more intermediate subgroup can be inserted that are distinct from the subgroups already in the series, then it is called a \emph{composition series}.  The factor groups of a composition series are called \emph{composition factors} and are \emph{simple groups}.  That is, each of their normal subgroups is either the whole group or is trivial.  One can equivalently define a composition series as a subnormal series in which all the factor groups are simple.  In a solvable group, all of the composition factors are cyclic so that there are no non-Abelian composition factors.

Much of the motivation for this work is based on a simple group called the \emph{projective special linear group}.  To obtain this group, one starts with the \emph{general linear group} $\gl_d(\bbF)$ of all invertible matrices over the field $\bbF$.  By restricting to the subgroup of matrices with determinant $1$, we obtain the \emph{special linear group} $\spl_d(\bbF)$.  The \emph{projective linear group} is then defined to be the quotient of $\spl_d(\bbF)$ mod the subgroup consisting of multiples of the identity matrix by roots of unity.  Dealing with the projective linear group is about as difficult as dealing with the general linear group since its non-Abelian composition factors consist of a single copy of the projective special linear group.

The \emph{holomorph} $\hol(G)$ of a group $G$ is a semidirect product of $G$ with its automorphism group.  An element $(g, \phi) \in \hol(G)$ acts on each element $x \in G$ by $(g, \phi)(x) = g \cdot \phi(x)$.  The product of two elements $(g_1, \phi_1), (g_2, \phi_2) \in \hol(G)$ is  $(g_1, \phi_1) \cdot (g_2, \phi_2) = (g_1 \phi_1(g_2), \phi_1 \phi_2)$.  The \emph{wreath product} $G \wr H$ of two permutation groups $G$ and $H$ that act on the sets $Y$ and $X$ is a semidirect product of the groups $G^H$ and $H$.  Here, $G^H$ means a direct product of $\abs{H}$ copies of $G$; each copy of $G$ is indexed by a different element of $H$.  Each element of the wreath product $G \wr H$ corresponds to a pair $(\bmg, h)$ where $\bmg \in G^H$ is a vector indexed by the elements of $H$.  Then $(\bmg, h)$ acts on an element $x, y \in X \times Y$ by $(g, \bmh)(x, y) = (g_{h y} x, h y)$.  Intuitively, a wreath product corresponds to a group of automorphisms of a full rooted tree of depth $2$.  The children of the root correspond to the elements of $Y$ while their children correspond to elements of $X \times Y$.  The element $h$ indicates how the children of the root should be permuted and the vector $\bmg$ of elements of $G$ indexed by $H$ indicates how the children of each child of the root should be permuted after the children of the root are permuted.  In particular, if $T$ is a rooted tree of depth $2$ where all nodes at depth $1$ have degree $d_1$ and all nodes at depth $d_2$, then $\aut(T) = S_{d_2} \wr S_{d_1}$.

In this section, we shall be concerned with \emph{iterated wreath products} of the form $G_1 \wr \cdots \wr G_k$ of groups $G_i$ which each acts on a set $X_i$.  The iterated wreath product $G_1 \wr \cdots \wr G_k$ acts on $X_1 \times \cdots \times X_k$ by recursively applying the definition of a wreath product.  If one imagines a rooted tree where the first level consists of the elements of $X_k$ and the \nth{i} level consists of the elements of $X_{k - i + 1} \times \cdots \times X_k$, then $G_1 \wr \cdots \wr G_k$ is a group of automorphisms of this rooted tree.  The group $G_k$ determines how the children of the root are permuted and there is a copy of each $G_{k - i + 1}$ for all $(x_{k - i + 1}, \ldots, x_k) \in X_{k - i + 1} \times \cdots \times X_1$ that determines how its children are permuted.  If one considers the full rooted tree of depth $k$ where every node in the $\nth{i}$ level has degree $d_i$, then its automorphism group is the iterated wreath product $S_{d_k} \wr \cdots \wr S_{d_1}$.  This notation quickly becomes cumbersome to deal with as the number of groups $k$ increases.  We address this problem by introducing a new notation for wreath products that is much more convenient for our purposes in \secref{red-group-iso}.

\subsection{Previous work on composition series isomorphism}
\label{ssec:prev-comp}
Our techniques rely on recent ideas by Luks~\cite{luks2015a} and Babai (personal communication) on composition series isomorphism.  We say that two series $S$ and $S'$ for groups $G$ and $H$ are isomorphic if there is an isomorphism from $G$ to $H$ that maps each subgroup in $S$ to the corresponding subgroup in $S'$.  Babai showed (personal communication) that if $\phi$ is an isomorphism between subnormal series\footnote{Only the case of composition series is relevant in this paper; however, Babai's result applies more generally to arbitrary subnormal series.} $S$ and $S'$ for groups $G$ and $H$, then $\phi \in \hol(F_0) \wr \cdots \wr \hol(F_k)$ where $F_0, \ldots, F_k$ are the factors of the isomorphic subnormal series $S$ and $S'$.  If $G$ and $H$ are solvable, then so is each $\hol(F_i)$; this implies that $\hol(F_0) \wr \cdots \wr \hol(F_k)$ is also solvable.  Since $\hol(F_0) \wr \cdots \wr \hol(F_k)$ can be given as a permutation group with $2 (k + 1)$ generators and color isomorphism problems on solvable groups can be handled in polynomial time~\cite{palfy1982a,babai1983a}, this implies that testing isomorphism of composition series of solvable groups is in polynomial time since $\hol(F_0) \wr \cdots \wr \hol(F_k)$ does not have any non-Abelian composition for solvable groups.

The \emph{solvable radical} $\rad(G)$ of a group $G$ is its unique maximal solvable normal subgroup.  Babai (personal communication) further proved that one can decide isomorphism of subnormal series of arbitrary groups in $n^{O(\log \log n)}$ time by using the algorithm of~\cite{babai2011a} assuming that they have the form $1 \tril G_1 \tril \cdots \tril G_k \tril \rad(G) \tril G_{k + 1} \tril \cdots \tril G_m = G$.  Using different but related ideas, Luks' went further and showed~\cite{luks2015a} that testing isomorphism of arbitrary composition series can be done in polynomial time.  In an upcoming paper (cf.~\cite{luks2015a}), Luks' plans to build this into the stronger result that canonical forms of composition series can be computed in polynomial time.  Let $p$ be the smallest prime divisor of the group.  Since every group has at most $n^{(1 / 2) \log_p n + O(1)}$ composition series, this method can be combined with the bidirectional collision detection methods introduced by the second author~\cite{rosenbaum2013b} (which provide a deterministic square-root speedup) to solve group isomorphism in $n^{(1 / 4) \log_p n + O(1)}$ time.

In this work, we apply Babai's holomorph trick and his idea to a different series that we call the \emph{radical derived series}\footnote{A similar series with elementary Abelian factors appears in~\cite{cannon2003a}.} (which we shall define shortly).  Unlike the classes of composition series and subnormal series, this series has the property that there is only one way to construct it for a given group.  Consequently, if $S$ and $S'$ denote the radical derived series for the groups $G$ and $H$, then $G$ and $H$ are isomorphic \ifft $S$ and $S'$ are isomorphic.  The advantage of this approach is that it allows us to avoid the $n^{(1 / 4) \log_p n}$ factor in the runtime above.  The difficulty of the group isomorphism problem is instead handled by allowing projective special linear composition factors in the resulting color isomorphism problem.  The radical derived series is defined as follows.
\begin{definition}
  Let $G$ be a group.  Then the \emph{radical derived series} of $G$ is
  \begin{equation*}
    \rad(G)^{(m)} = 1 \tril \rad(G)^{(m - 1)} \tril \cdots \tril \rad(G)^{(0)} = \rad(G) \tril G
  \end{equation*}
\end{definition}
Here, $\rad(G)^{(i)}$ denotes the \nth{i} subgroup in the derived series of $\rad(G)$ starting with $\rad(G)$.

Because the iterated wreath products that arise in this reduction are quite complicated and difficult to handle, we also introduce a new notation for describing elements of iterated wreath products which makes our proofs much easier.  It is our hope that our notation will prove useful in future work in this area.  Our proof also requires us to prove a result on the composition factors of the automorphism groups of Abelian groups.  We accomplish this by using the framework for dealing with automorphisms of Abelian groups given in~\cite{hillar2006a}.
% The advantage The idea here is to handle the difficult part of the group isomorphism problem by allowing more complex composition factors in the resulting color isomorphism problem rather than considering many
% While a group of order $n$ can have as many as $n^{(1 / 2) \log_p n}$ composition series where $p$ is the smallest prime that divides the order of the group, every group has a unique radical derived series.  By reducing to the radical derived series, so this allows us to avoid the $n^{(1 / 4) \log_p n}$ factor above.  However, the group that we obtain has non-Abelian composition factors, so we obtain a more difficult color isomorphism problem.  We rely on both of Babai's holomorph trick and on his method for generalizing from solvable groups to arbitrary groups.

\section{Reducing group isomorphism to color isomorphism}
\label{sec:red-group-iso}
In this section, we prove that group isomorphism reduces to the color isomorphism problem with cyclic and projective special linear composition factors.

The first step is to identify the elements of $G$ and $H$.  Clearly, this does not solve the isomorphism problem since the resulting groups can have different multiplication rules and the identification does not necessarily yield an isomorphism.  We accomplish this using the following definitions.

\begin{definition} %$G^{(m)} = 1 \tril G^{(m - 1)} \tril \cdots \tril G^{(0)} = G$
  \label{defn:hat-G}
  Let $G$ be a group, let $G_m = 1 \tril \cdots \tril G_0 = G$ be its radical derived series and let $F_i = G_{m - i} / G_{m - i + 1}$ and choose an arbitrary lift $\ell_i : F_i \ra G_i$ for each $1 \leq i \leq m$.  Then for each $g \in G$, there exists a unique $(x_1, \ldots, x_m) \in F_1 \cart \cdots \cart F_m$ such that $g = \ell_m(x_0) \cdots \ell_1(x_{m - 1})$.  Let  $\ell : F_1 \cart \cdots \cart F_m \ra G$ denote this bijection.  Then we define $\hat G$ to be the group on the set $F_1 \cart \cdots \cart F_m$ whose multiplication rule is induced by $G$ under the bijection $\ell$.
\end{definition}

Here, we distinguish between Cartesian products (denoted by $\cart$) which operate on sets and direct products (denoted by $\times$) which operate on groups.  Therefore, $F_1 \cart \cdots \cart F_m$ is the subset of $\hat G$ that corresponds to $(\hat G)^{(i)}$.  We use $\cart$ instead of $\times$ in order to avoid suggesting that $F_1 \cart \cdots \cart F_m$ is a direct product of the groups $F_1, \ldots, F_m$.  This would be very misleading since $(\hat G)^{(i)}$ can be non-Abelian. % It is not necessarily a direct product of groups and refers merely to the subgroup of $\tilde G$ defined on this set (which may be non-Abelian).

Our next step is to identify the factor groups in the radical derived series for the groups $G$ and $H$.  Let us say that the \emph{canonical representation} of an Abelian group $A$ is the unique group of the form $\cC(A) = \bigtimes_{i = 1}^k \bbZ_{p_i}^{e_i}$ that is isomorphic to $A$ where $p_1 < \cdots < p_k$ are primes and each $e_i$ is a natural number.  This takes care of the Abelian factors in the radical derived series.  However, we also need a way to identify the factor groups $G / \rad(G)$ and $H / \rad(H)$.  These are non-Abelian groups that do not have any normal Abelian subgroups.  We identify them by using the algorithm of~\cite{babai2011a} which can enumerate all the isomorphisms between two groups of order $n$ that do not have any normal Abelian subgroups in $n^{O(\log \log n)}$ time.  To do this, we first define $\cC(G / \rad(G)) = G / \rad(G)$; we then define $\cC(H / \rad(H)) = G / \rad(G)$.  The later identification is performed in $n^{O(\log \log n)}$ time using~\cite{babai2011a}.  Note that our definition of $\cC$ on non-Abelian groups depends on whether we are given $G$ or $H$ and is thus specific to our problem instance.
%For a non-Abelian simple group $S$, we let $\cC(S)$ be some canonical way of writing the group $S$.  We note that $\cC(S)$ can be computed in polynomial time since every non-Abelian simple group has a generating set of size two~\cite{malle1994a}.

\begin{definition} % Let $G$ be a group, let $G^{(m)} = 1 \tril G^{(m - 1)} \tril \cdots \tril G^{(0)} = G$ be its derived series and let $F_i = G^{(m - i)} / G^{(m - i + 1)}$
  \label{defn:tilde-G}
  Let $G$ be a group, let $G_m = 1 \tril \cdots \tril G_0 = G$ be its radical derived series and let $F_i = G_{m - i} / G_{m - i + 1}$.  For each $F_i$, let $\tilde F_i = \cC(F_i)$ and choose an arbitrary isomorphism $\varphi_i : F_i \ra \tilde F_i$.  This defines a bijection $\varphi : F_1 \cart \cdots \cart F_m \ra \tilde F_1 \cart \cdots \cart \tilde F_m$.  We let $\tilde G$ be the group on the set $F_1 \cart \cdots \cart \tilde F_m$ whose multiplication rule is induced by the group $\hat G$ under the bijection $\varphi : \hat G \ra \tilde G$.
\end{definition}

Note that $\varphi \circ \ell^{-1} : G \ra \tilde G$ is an isomorphism from $G$ to $\tilde G$.  A key fact that we shall need about $\tilde G$ is that its derived subgroups correspond to iteratively removing factors from the product $\cC(F_1) \cart \cdots \cart \cC(F_m)$ as we move down the series.  This is stated in the following proposition.  The proof follows easily from the definitions.

\begin{proposition}
  \label{prop:tg-der} % Let $G$ be a group, let $G^{(m)} = 1 \tril G^{(m - 1)} \tril \cdots \tril G^{(0)} = G$ be its derived series and let $F_i = G^{(m - i)} / G^{(m - i + 1)}$.
  Let $G$ be a group, let $G_m = 1 \tril \cdots \tril G_0 = G$ be its radical derived series and let $F_i = G_{m - i} / G_{m - i + 1}$.  Then $(\tilde G)^{(i)} = \cC(F_1) \cart \cdots \cart \cC(F_{m - i})$.
\end{proposition}

The next lemma allows us to treat isomorphisms between groups as members of a wreath product.  %We make heavy use of Babai's holomorph idea in this lemma and its proof.

%Recall that the \nth{i} derived subgroup of $G$ is denoted by
\begin{lemma}
  \label{lem:iso-wr}
  Let $G$ and $H$ be groups and suppose that $\phi : \tilde G \ra \tilde H$ is an isomorphism.  Let $G_m = 1 \tril \cdots \tril G_0 = G$ and $H_m = 1 \tril \cdots \tril H_0 = H$ be the radical derived series for $G$ and $H$ and let $\tilde F_i = \cC(G_{m - i} / G_{m - i + 1}) = \cC(H_{m - i} / H_{m - i + 1})$.  Then $\phi \in \hol(\tilde F_1) \wr \cdots \wr \hol(\tilde F_m)$.
\end{lemma}

Before we can present the proof, we need a better way of dealing with iterated wreath products since using the standard wreath product definition recursively quickly becomes very cumbersome.  Without better notation, our proof would be extremely tedious.  We accomplish this by defining a wreath product as a indexes set of elements that satisfies certain conditions.  %We also use the notation $T(X_1, \ldots, X_k)$ as a shorthand for $\bigcup_{i = 1}^{k - 1} \bigtimes_{j = 1}^{i} X_i$.

\begin{definition}
  \label{defn:wr-ten}
  Consider the iterated wreath product $G_1 \wr \cdots \wr G_k$ of groups $G_i$ which each acts on a set $X_i$.  Let $\pi_{x_{i + 1}, \ldots, x_k} \in G_i$ for each $x_{i + 1}, \ldots, x_k \in X_{i + 1} \times \cdots \times X_k$ and each $1 \leq i \leq k$.  Then this set of elements defines the permutation $\pi(x_1, \ldots, x_k) = (\pi_{x_2, \ldots, x_k}(x_1), \ldots, \pi_{x_k}(x_{k - 1}), \pi_{()}(x_k))$.
%Then $\pi : \bigtimes_{i = 1}^{k } X_i \ra \bigtimes_{i = 1}^k G_i$ denotes an element of $G_1 \wr \cdots \wr G_k$ if, for all $(x_1, \ldots, x_k), (x_1', \ldots, x_k') \in \bigtimes_{i = 1}^k X_i$ where $x_i = x_i'$ for $i \leq j$, $(\pi_{x_1, \ldots, x_k})_i = (\pi_{x_1', \ldots, x_k'})_i$ for $i \leq j$.  Interpreted as a permutation, $\pi$ acts on $\bigtimes_{i = 1}^k X_i$ by $\pi(x_1, \ldots, x_k) = ()$
\end{definition}

Note that in the above definition, $\pi_{()} \in G_k$ denotes the case where $i = k + 1$ so that the list of subscripts is empty.  It is easy to show that the functions $\pi$ from \defref{wr-ten} are indeed permutations and are precisely the elements of the iterated wreath product $G_1 \wr \cdots \wr G_k$.

\begin{lemma}
  \label{lem:wr-ten}
  Consider the iterated wreath product $G_1 \wr \cdots \wr G_k$ of groups $G_i$ which each acts on a set $X_i$.  Then every $\pi$ defined by \defref{wr-ten} is a permutation contained in $G_1 \wr \cdots \wr G_k$.  Moreover, every element of $G_1 \wr \cdots \wr G_k$ can be expressed in the form of \defref{wr-ten}.
\end{lemma}

% TODO: Add proof?

%For each $(x_{i - 1}, \ldots, x_1)$ with $0 \leq i \leq k$ , let $\sigma_{x_{i - 1}, \ldots, x_1} \in G_i$.  We then define $\pi$ to be the permutation of $X_1 \times \cdots \times X_k$ where $\pi(x_k, \ldots, x_1) = (\sigma_{x_{k - 1}, \ldots, x_1} (x_k), \sigma_{x_{k - 2}, \ldots, x_1} (x_{k - 1}), \ldots, \sigma_{()} (x_1))$ for each $(x_k, \ldots, x_1)$.  It is easy to show that each such $\pi$ is a permutation as claimed and the elements of $G_1 \wr \cdots \wr G_k$ correspond precisely to the permutations that can be defined in this way.

Now, we are ready to prove \lemref{iso-wr}.  Because our proof deals with many subsequences of vectors, we introduce a shorthand.  If $\bma \in A_1 \times \cdots \times A_k$, then $\bma_{i, j}$ denotes the subsequence $(a_i, \ldots, a_j)$.

\begin{proof}
  For each $1 \leq i \leq m$, let $\phi_i = \restr{\phi}{F_i} : \tilde F_i \ra \tilde F_i$ be the automorphism of $\tilde F_i$ induced by $\phi$.  Define
  \begin{equation*}
    \tilde \phi_{\bmx_{i + 1, m}} = (f_i(\bmx_{i + 1, m}), \phi_i) \in \hol(\tilde F_i)
  \end{equation*}
  for each $\bmx \in F_1 \cart \cdots \cart F_m$ and $1 \leq i \leq m$ where the $f_{i + 1} : F_{i + 1} \cart \cdots \cart F_m \ra F_i$ are functions that are to be determined in the course of the proof.  Whatever we later choose these functions to be, note that by \defref{wr-ten} and \lemref{wr-ten}, it defines a permutation $\tilde \phi \in \hol(\tilde F_1) \wr \cdots \wr \hol(\tilde F_m)$.  Our aim is to choose them so that $\phi = \tilde \phi$.

  We accomplish this by induction on $i$.  Our goal is to show that
  \begin{equation}
    \label{eq:iso-wr-ind}
    \phi(\bone^i, \bmx_{i + 1, m}) = \tilde \phi(\bone^i, \bmx_{i + 1, m})
  \end{equation}
  for all $0 \leq i \leq m$ where $\bone^i$ is a shorthand for $\overbrace{1, \ldots, 1}^{\text{$i$ times}}$.

  We start with the basis case $i = m$.  This corresponds to the claim that
  \begin{equation*}
    \phi(\bone^m) = \tilde \phi(\bone^m)
  \end{equation*}
  which is equivalent to asserting that
  \begin{equation*}
    \bone^m = (f_1(\bone^{m - 1}), \ldots, f_{m - 1}(1), f_m)
  \end{equation*}
  Since we can choose the $f_i$ functions as desired, we simply define $(f_1(\bone^{m - 1}), \ldots, f_{m - 1}(1), f_m) = \bone^m$.  This proves the basis case, so we now proceed to the inductive case.

  Assume that \eqref{iso-wr-ind} holds for some $1 \leq i \leq m$; we will show that it holds for $i - 1$ as well.  By \propref{tg-der}, $(\bone^{i - 1}, x_i, \bone^{m - i}) \in (\tilde G)^{(m - i)}$ so
  \begin{equation*}
    \phi(\bone^{i - 1}, \bmx_{i, m}) = \phi(\bone^{i - 1}, x_i, \bone^{m - i}) \phi(\bone^i, \bmx_{i + 1, m})
  \end{equation*}
  Now, again because $\phi(\bone^{i - 1}, x_i, \bone^{m - i}) \in (\tilde G)^{(m - i)}$, we have $\phi(\bone^{i - 1}, x_i, \bone^{m - i}) = (\bma_{i - 1}(x_i), \phi_i(x_i), \bone^{m - i})$ for some $\bma_{i - 1}(x_i) \in \tilde F_1 \cart \cdots \cart \tilde F_{i - 1}$.  Thus, by the inductive hypothesis
  \begin{equation}
    \label{eq:iso-wr-ind-1}
    \phi(\bone^{i - 1}, \bmx_{i, m}) = (\bma_{i - 1}(x_i), \phi_i(x_i), \bone^{m - i}) \cdot \tilde \phi(\bone^i, \bmx_{i + 1, m})
  \end{equation}
  By \defref{wr-ten} and noting that $f_m = 1$ from the basis case, we see that $\tilde \phi(\bone^i, \bmx_{i + 1, m})$ is equal to
  \begin{equation*}
    (f_1(\bone^{i - 1}, \bmx_{i + 1, m}), \ldots, f_i(\bmx_{i + 1, m}), f_{i + 1}(\bmx_{i + 2, m}) \phi_i(x_{i + 1}), \ldots, f_{m - 1}(x_m) \phi_i(x_{m - 1}), \phi_m(x_m))
  \end{equation*}
  By applying \propref{tg-der}, we see that this is equal to
  \begin{align}
    \label{eq:iso-wr-ind-2}
    {} & (f_1(\bone^{i - 1}, \bmx_{i + 1, m}), \ldots, f_i(\bmx_{i + 1, m}), \bone^{m - i}) \\
    {} & \cdot (\bone^i, f_{i + 1}(\bmx_{i + 2, m}) \phi_i(x_{i + 1}), \ldots, f_{m - 1}(x_m) \phi_i(x_{m - 1}), \phi_m(x_m)) \nonumber
  \end{align}
  By replacing $\tilde \phi(\bone^i, \bmx_{i + 1, m})$ in \eqref{iso-wr-ind-1} with \eqref{iso-wr-ind-2}, we see that
  \begin{align*}
    \phi(\bone^{i - 1}, \bmx_{i, m}) = {} & (\bma_{i - 1}(x_i), \phi_i(x_i), \bone^{m - i}) \\
    {} & \cdot (f_1(\bone^{i - 1}, \bmx_{i + 1, m}), \ldots, f_i(\bmx_{i + 1, m}), \bone^{m - i}) \\
    {} & \cdot (\bone^i, f_{i + 1}(\bmx_{i + 2, m}) \phi_i(x_{i + 1}), \ldots, f_{m - 1}(x_m) \phi_i(x_{m - 1}), \phi_m(x_m))
  \end{align*}
  Now $\tilde F_i$ is Abelian for $i < m$ and $f_i = 1$ for $i = m$.  Thus, $\phi_i(x_i) f_i(\bmx_{i + 1, m}) = f_i(\bmx_{i + 1, m}) \phi_i(x_i)$, so
  \begin{align*}
    \phi(\bone^{i - 1}, \bmx_{i, m}) = {} & (\bmb_{i - 1}(B), f_i(\bmx_{i + 1, m}) \phi_i(x_i), \bone^{m - i}) \\
    {} & \cdot (\bone^i, f_{i + 1}(\bmx_{i + 2, m}) \phi_i(x_{i + 1}), \ldots, f_{m - 1}(x_m) \phi_i(x_{m - 1}), \phi_m(x_m)) \\
    {} = {} & (\bmb_{i - 1}(B), f_i(\bmx_{i + 1, m}) \phi_i(x_i), \ldots, f_{m - 1}(x_m) \phi_i(x_{m - 1}), \phi_m(x_m))
  \end{align*}
  where $B = (x_i, f_1(\bone^{i - 1}, \bmx_{i + 1, m}), \ldots, f_i(\bmx_{i + 1, m}))$ indicates the values on which $\bmb_{i - 1}(B)$ depends\footnote{There is no need to include $\phi_i(x_i)$ in addition to $x_i$ since $\phi$ is fixed and therefore so is $\phi_i$.}.

  If $i = 1$, then $\bmb_{i - 1}(B) = \bmb_0(B) = ()$ and we have $\phi(\bone^{i - 1}, \bmx_{i, m}) = \tilde \phi(\bone^{i - 1}, \bmx_{i, m})$ as desired.  Also, if $x_i = 1$, then $\phi(\bone^{i - 1}, \bmx_{i, m}) = \tilde \phi(\bone^{i - 1}, \bmx_{i, m})$ by the inductive hypothesis.  Therefore, we may assume that $i \geq 2$ and $x_i \not= 1$.  Then by \defref{wr-ten} we have
  \begin{equation*}
    \tilde \phi(\bone^{i - 1}, \bmx_{i, m}) = (f_1(\bone^{i - 2}, \bmx_{i, m}), \ldots, f_{i - 1}(\bmx_{i, m}), f_i(\bmx_{i + 1, m}) \phi_i(x_i), \ldots, f_{m - 1}(x_m) \phi_{m - 1}(x_{m - 1}), \phi_m(x_m))
  \end{equation*}
  so to show that $\phi(\bone^{i - 1}, \bmx_{i, m}) = \tilde \phi(\bone^{i - 1}, \bmx_{i, m})$, we need to choose $(f_1(\bone^{i - 2}, \bmx_{i, m}), \ldots, f_{i - 1}(\bmx_{i, m})) = \bmb_{i - 1}(B)$.

  Now, we just need to argue that all the assignments that we make at each step are independent.  Let us say that the weight of a vector $\bmy$ is $d - k$ where $d$ is the length of $\bmy$ and $k$ is the smallest index such that $y_k \not= 1$.  Then, the \nth{i} step of the induction assigns values to the functions $f_j$ with $j < i$ on arguments of weight exactly $m - i + 1$.  It follows that the assignments made at each step are independent which proves that \eqref{iso-wr-ind} holds for $i - 1$.  By induction, we conclude that the functions $f_i$ can be chosen so that $\phi = \tilde \phi$.
\end{proof}

Before we can prove \thmref{red-group}, we need a lemma about the structure of automorphisms of Abelian groups.

\begin{lemma}
  \label{lem:abel-aut}
  Let $A$ be an Abelian group.  Then every composition factor of $\aut(A)$ is either cyclic or projective special linear.
\end{lemma}

To prove this lemma, we need to introduce a few definitions and results on the theory of automorphisms of Abelian groups.  These were first studied by Ranum~\cite{ranum1907a}; however, we follow the more modern treatment by Hillar and Rhea~\cite{hillar2006a} since it is more convenient.  Since the group of automorphisms of a direct product of groups of relatively prime order is the direct product of the automorphisms of each group, it suffices to consider Abelian $p$-groups.  First, we characterize the endomorphisms of Abelian $p$-groups.

\begin{definition}[\cite{ranum1907a}, cf.~\cite{hillar2006a}]
  \label{defn:abel-R}
  Let $A = \bigtimes_{i = 1}^d \bbZ_p^{e_i}$ be a Abelian $p$-group where $e_1 < \cdots < e_d$.  Define
  \begin{equation*}
    R(A) = \setb{(m_{ij}) \in \bbZ^{d \times d}}{\text{$p^{e_i - e_j}$ divides $a_{ij}$ for all $1 \leq j \leq i \leq d$}}
  \end{equation*}
\end{definition}

One can show that $R(A)$ is a ring~\cite{ranum1907a} (cf.~\cite{hillar2006a}).  The endomorphisms $\edm(A)$ of $A$ then arise via a homomorphism defined on $R(A)$.

\begin{theorem}[\cite{hillar2006a}]
  \label{thm:abel-R-end}
  Let $A = \bigtimes_{i = 1}^d \bbZ_p^{e_i}$ be a Abelian $p$-group where $e_1 < \cdots < e_d$ and define $\psi : R(A) \ra \edm(A)$ by
  \begin{equation*}
    \psi(M)(\pi(\bma)) = \pi(M \bma)
  \end{equation*}
  where $\pi : \bbZ^d \ra A$ is the projection that maps each $\bma \in \bbZ^d$ to $(a_1 + \bbZ_p^{e_1}, \ldots, a_d + \bbZ_p^{e_d})$.  Then $\psi$ is a surjective homomorphism.
\end{theorem}

We also need another result that relates the endomorphisms to automorphisms.

\begin{theorem}[\cite{ranum1907a}, cf.~\cite{hillar2006a}]
  \label{thm:abel-end-aut}
  Let $A = \bigtimes_{i = 1}^d \bbZ_p^{e_i}$ be an Abelian $p$-group where $e_1 < \cdots < e_d$.  Then $\psi(M)$ is an automorphism \ifft $\psi(M) \bmod p \in \gl_d(p)$ (where the modulo division is performed entrywise).
\end{theorem}

Now, we are ready to prove \lemref{abel-aut}.

\begin{proof}[Proof of \lemref{abel-aut}]
  Since the automorphism group of $A$ is the direct product of the automorphism groups of its Sylow subgroups, it suffices to prove this for the case where $A$ is a $p$-group.  Let $A = \bigtimes_{i = 1}^d \bbZ_p^{e_i}$ be a Abelian $p$-group where $e_1 < \cdots < e_d$ and let us define $\rho : \edm(A) \ra \bbZ_p^{d \times d}$ by $\rho(\psi(M)) = \psi(M) \bmod p$ for each $\psi(M) \in \edm(A)$.  Observe that $\rho$ is a ring homomorphism.

  By \defref{abel-R} and \thmref{abel-R-end}, the image of of $\rho$ is
  \begin{equation*}
    \im \rho = \setb{(m_{ij}) \in \bbZ_p^{d \times d}}{m_{ij} = 0 \text{ if $1 \leq j \leq i \leq d$ and $e_i \not= e_j$}}
  \end{equation*}
  In other words, the image of $\rho$ consists of block-upper triangular matrices in $\bbZ_p^{d \times d}$ where the blocks consist of those $(i, j)$ such that $e_i = e_j$.  Let $B_k$ denote the set of indexes $(i, j)$ in the \nth{k} block on the main diagonal of these matrices where $1 \leq k \leq \ell$.  Since the determinant of a block-upper triangular matrix is equal to the product of the determinants of the blocks, we see that $\rho[\aut(A)]$ consists of those matrices in $\im \rho$ where the blocks on the diagonal are invertible.  Thus,
  \begin{equation*}
    \rho[\aut(A)] = \setb{(m_{ij}) \in \im \rho}{\det\left[(m_{ij})^{B_k}\right] \not= 0 \text{ for each } k}
  \end{equation*}
  where $(m_{ij})^{B_k}$ denotes the submatrix of $(m_{ij})$ on the block $B_k$.

  We now shift our attention to $\sigma = \restr{\rho}{\aut(A)} : \aut(A) \ra \rho[\aut(A)]$ which we interpret as a surjective homomorphism between multiplicative groups.  Since $\aut(A) / \ker \sigma \cong \im \sigma$, to find the composition factors of $\aut(A)$, it suffices to show that the composition factors of the kernel and image of $\rho$ are either cyclic or projective special linear.  Now, $\im \sigma = \rho[\aut(A)]$.

  To find its composition factors, we define another homomorphism $\theta : \im \sigma \ra \gl_d(p)$ where $d_k$ is the dimension of the \nth{k} block $B_k$.  Then we define $\theta(M) = \diag(M_1, \ldots, M_\ell)$ where $M_k = (m_{ij})^{B_k}$ is the submatrix on the block $B_k$ in $M$.  Now, $\im \theta = \gl_{d_1}(p) \times \cdots \times \gl_{d_\ell}(p)$ and the composition factors of general linear groups are cyclic and projective special linear.  The kernel of $\theta$ is a $p$-group and therefore has cyclic composition factors.  It follows that $\im \sigma$ has only cyclic and projective special linear composition factors.

  All the remains is to determine the composition factors of $\ker \sigma$.  However, $\ker \sigma$ is also a $p$-group, so its composition factors are all cyclic.  It follows that the composition factors of $\aut(A)$ are cyclic and projective special linear.
\end{proof}

\redgroup*

\begin{proof}
  % Let $G$ and $H$ be groups and suppose that $\phi : \tilde G \ra \tilde H$ is an isomorphism.
  % Then $\phi \in \hol(\tilde F_1) \wr \cdots \wr \hol(\tilde F_m)$.
  Let $G_m = 1 \tril \cdots \tril G_0 = G$ and $H_m = 1 \tril \cdots \tril H_0 = H$ be the radical derived series for two groups $G$ and $H$ and let $\tilde F_i = \cC(G_{m - i} / G_{m - i + 1}) = \cC(H_{m - i} / H_{m - i + 1})$.  We will first compute the isomorphisms $\iso(\tilde G, \tilde H)$ from $\tilde G$ to $\tilde H$ and then use this to find the isomorphisms $\iso(G, H)$ from $G$ to $H$.  If $G$ and $H$ are isomorphic, then \lemref{iso-wr} tells us that $\iso(\tilde G, \tilde H) \subseteq \hol(\tilde F_1) \wr \cdots \wr \hol(\tilde F_m)$.  Now, by the definition of the holomorph and \lemref{abel-aut}, the composition factors of $\hol(\tilde F_i)$ are either cyclic or projective special linear for $1 \leq i < m$.  Therefore, the composition factors of the normal subgroup $\left(\hol(\tilde F_1) \wr \cdots \wr \hol(\tilde F_{m - 1})\right)^{\tilde F_m}$ are all either cyclic or projective special linear.  (As before, $\left(\hol(\tilde F_1) \wr \cdots \wr \hol(\tilde F_{m - 1})\right)^{\tilde F_m}$ denotes a direct product of copies of $\hol(\tilde F_1) \wr \cdots \wr \hol(\tilde F_{m - 1})$ indexed by $\tilde F_m$.)  Since $\tilde F_m$ can be non-Abelian, $\hol(\tilde F_m)$ can have other composition factors which we must somehow eliminate if we are to place \GRI\spc in \CIS.

  %We accomplish this using the same trick that Babai utilized to extend his algorithm for testing isomorphism of composition series of solvable groups to arbitrary composition series.  Namely,
  We accomplish this using the results of~\cite{babai2011a}, which show that a group of order $n$ that does not have any Abelian normal subgroups has at most $n^{O(\log \log n)}$ automorphisms and that all of them can be enumerated within the same bound.  This implies that $\abs{\aut(\tilde F_m)} \leq \abs{\tilde F_m}^{O(\log \log \abs{\tilde F_m})} \leq n^{O(\log \log n)}$ where $n = \abs{G} = \abs{H}$ and that we can enumerate $\aut(\tilde F_m)$ within the same bound.  Consequently, we can also enumerate $\hol(\tilde F_m)$ in $n^{O(\log \log n)}$ time.

  Therefore, we can transform our instance of \GRI\spc into $n^{O(\log \log n)}$ instances of \CIS as follows.  For each $(f_m, \phi_m) \in \hol(\tilde F_m)$, we consider the coset $\left(\hol(\tilde F_1) \wr \cdots \wr \hol(\tilde F_{m - 1})\right)^{\tilde F_m} \cdot (f_m, \phi_m)$.  Note that
  \begin{equation*}
    \hol(\tilde F_1) \wr \cdots \wr \hol(\tilde F_m) = \bigcup_{(f_m, \phi_m) \in \hol(\tilde F_m)} \left(\hol(\tilde F_1) \wr \cdots \wr \hol(\tilde F_{m - 1})\right)^{\tilde F_m} \cdot (f_m, \phi_m)
  \end{equation*}
  so it suffices to find the isomorphisms from $\tilde G$ to $\tilde H$ that are contained in each such coset and accumulate the results.  For each $\left(\hol(\tilde F_1) \wr \cdots \wr \hol(\tilde F_{m - 1})\right)^{\tilde F_m} \cdot (f_m, \phi_m)$, we define a color isomorphism problem where we extend $\left(\hol(\tilde F_1) \wr \cdots \wr \hol(\tilde F_{m - 1})\right)^{\tilde F_m} \cdot (f_m, \phi_m)$ to act on the set $X = \left(\tilde F_1 \cart \cdots \cart \tilde F_m\right)^3$.  Recalling that $\tilde F_1 \cart \cdots \cart \tilde F_m$ is the underlying set of both $\tilde G$ and $\tilde H$, we solve the instance of \CIS\spc that arises when we let $f_1 : X \ra [n^3]$ and $f_2 : X \ra [n^3]$ be the indicator functions on the subsets $\setb{(x, y, xy)}{x, y \in \tilde G}$ and $\setb{(x, y, xy)}{x, y \in \tilde H}$ of $X$.  This yields all isomorphisms from $\tilde G$ to $\tilde H$ that are contained in the coset $\left(\hol(\tilde F_1) \wr \cdots \wr \hol(\tilde F_{m - 1})\right)^{\tilde F_m} \cdot (f_m, \phi_m)$.  By taking the union of all of the isomorphisms found, we obtain $\iso(\tilde G, \tilde H)$ in $n^{O(\log \log n)}$ time.

  All that remains is to show how to compute $\iso(G, H)$ from $\iso(\tilde G, \tilde H)$.  Since it was computed from $n^{O(\log \log n)}$ cosets, the description of $\iso(\tilde G, \tilde H)$ may use up to $n^{O(\log \log n)}$ generators.  For convenience, we reduce this to $O(\log^2 n)$ generators in polynomial time using standard permutation group algorithms (cf.~\cite{seress2003a}).  Using Definitions~\ref{defn:hat-G} and~\ref{defn:tilde-G}, we can define isomorphisms $\alpha : G \ra \tilde G$ and $\beta : H \ra \tilde H$.  Then $\iso(G, H) = \beta^{-1} \iso(\tilde G, \tilde H) \alpha$.
\end{proof}

\section{Reducing color isomorphism to a generalization of group isomorphism}
\label{sec:red-color-iso}
In this section, we show that the problem of computing the isometry group of a bilinear map can be reduced to a generalization of group isomorphism in polynomial time.  The first step is to construct a group whose structure depends on the bilinear map.

\begin{definition}
  Let $f : B \ra A \times A$ be a bilinear map.  Then we define $G_f$ to be the group on the set $B \times A$ with the operation $(b_1, a_1) \cdot (b_2, a_2) = (b_1 b_2, a_1 a_2 f(b_1, b_2))$.
\end{definition}

The fact that $G_f$ is a group follows easily from the assumption that $f$ is bilinear.  Readers familiar with group cohomology theory will note that this is also a consequence of a construction from group cohomology involving factor sets (cf.~\cite{rotman1995a}).  Moreover, since $A \leq Z(G)$, $G_f / Z(G_f)$ is isomorphic to a subgroup of the Abelian group $B$, so it follows that $G_f$ is nilpotent of class at most $2$.

\begin{proposition}
  \label{prop:G-2-nil-group}
  $G_f$ is a nilpotent group of class at most $2$.
\end{proposition}

Let $b \in B$.  It is convenient to define $\ell(b) = (b, 0)$.  We can then write $\ell(b) a$ for $(b, a)$. %In this form, the multiplication rule of $G_f$ becomes $(\ell(b_1) a_1) (\ell(b_2) a_2) = \ell(b_1 b_2) a_1 a_2$.

% The following proposition gives several basic but useful properties of the group $G_f$.
%
% \begin{proposition}
%   Let $f : B \ra A \times A$ be a bilinear map.  Then
%   \begin{enumerate}
%   \item $A \tril G_f$
%   \item $A \leq Z(G_f)$
%   \end{enumerate}
% \end{proposition}

Our next step is to show that every isometry of $f$ gives rise to an automorphism of $G_f$.

\begin{lemma}
  \label{lem:isom-aut}
  Let $\beta : B \ra B$ be an isometry of a bilinear map $f : B \times B \ra A$.  Then the map $\phi : G_f \ra G_f$ defined by $\phi(\ell(b) a) = \ell(\beta b) a$ is an automorphism of $G_f$.
\end{lemma}

\begin{proof}
  Let $\ell(b_1) a_1, \ell(b_2) a_2 \in G_f$.  Then
  \begin{align*}
    \phi((\ell(b_1) a_1) (\ell(b_2) a_2)) & = \phi(\ell(b_1 b_2) f(b_1, b_2) a_1 a_2) \\
                                       {} & = \ell(\beta (b_1 b_2)) f(b_1, b_2) a_1 a_2 \\
                                       {} & = \ell(\beta (b_1)) \ell(\beta(b_2)) f(\beta b_1, \beta b_2)^{-1} f(b_1, b_2) a_1 a_2 \\
                                       {} & = \ell(\beta (b_1)) a_1 \ell(\beta(b_2)) a_2 \\
                                       {} & = \phi(\ell(b_1) a_1) \phi(\ell(b_2) a_2)
  \end{align*}
  % Let $(b_1, a_1), (b_2, a_2) \in G_f$.  Then% $f(b_1, b_2) = f(\beta b_1, \beta b_2)$.  Therefore,
  % \begin{align*}
  %   \phi((b_1, a_1) \cdot (b_2, a_2)) & = \phi((b_1 + b_2, a_1 + a_2 + f(b_1, b_2))) \\
  %                                  {} & = \phi((b_1 + b_2, 0)) \phi(0, a_1 + a_2 + f(b_1, b_2)) \\
  %                                  {} & = (\beta b_1 + \beta b_2, 0) \cdot (0, a_1 + a_2 + f(b_1, b_2)) \\
  %                                  {} & = (\beta b_1, 0) \cdot (\beta b_2, 0) \cdot (0, a_1 + a_2 + f(b_1, b_2) - f(\beta b_1, \beta b_2)) \\
  %                                  {} & = (\beta b_1, a_1) \cdot (\beta b_2, a_2) \\
  %                                  {} & = \phi(b_1, a_1) \cdot \phi(b_2, a_2)
  % \end{align*}
  which completes the proof.
\end{proof}

We also need to show that certain types of automorphisms of $G_f$ yield isometries of $f$.  Our first step towards this goal is to prove the following characterization of which automorphisms of $G_f$ that fix $A$ induce isometries.  For convenience, we identify $B$ with $G_f / A$ via $b \mapsto \ell(b) A$.  An automorphism $\phi$ of $G_f$ can then induces an automorphism $\beta = \restr{\phi}{B}$ of $B$ by taking images of the cosets $\ell(b) A$.

\begin{lemma}
  \label{lem:aut-hom}
  Let $f : B \times B \ra A$ be a bilinear map and let $\phi \in \aut(G_f)$ such that $\phi[A] = A$.  Let $\beta = \restr{\phi}{B} : B \ra B$ and define $\ell_{\phi} = \phi \ell \beta^{-1} : B \ra G_f$ and $\varphi_{\phi} : B \ra A$ by $\varphi_{\phi}(b) = \ell(b) \ell_{\phi}(b)^{-1}$.  Then $\varphi_{\phi} \in \hom(B, A)$ \ifft $f(b_1, b_2) = f(\beta^{-1} b_1, \beta^{-1} b_2)$ for all $b_1, b_2 \in B$.
\end{lemma}

\begin{proof}
  Let $b_1, b_2 \in B$.  Then
  \begin{align*}
    \varphi_{\phi}(b_1 b_2) & = \ell(b_1 b_2) \left[ \phi \ell(\beta^{-1}(b_1 b_2)) \right]^{-1}\\
                         {} & = \ell(b_1) \ell(b_2) f(b_1, b_2)^{-1} \left[ \phi\left[ \ell(\beta^{-1} b_1) \ell(\beta^{-1} b_2) f(\beta^{-1} b_1, \beta^{-1} b_2)^{-1} \right] \right]^{-1} \\
                         {} & = \ell(b_1) \ell(b_2) \ell_{\phi}(b_2)^{-1} \ell_{\phi}(b_1)^{-1} f(b_1, b_2)^{-1} f(\beta^{-1} b_1, \beta^{-1} b_2) \\
                         {} & = \ell(b_1) \ell_{\phi}(b_1)^{-1} \varphi_{\phi}(b_2) f(b_1, b_2)^{-1} f(\beta^{-1} b_1, \beta^{-1} b_2) \\
                         {} & = \varphi_{\phi}(b_1) \varphi_{\phi}(b_2) f(b_1, b_2)^{-1} f(\beta^{-1} b_1, \beta^{-1} b_2)
  \end{align*}
  Now, $\varphi_{\phi}$ is a homomorphism \ifft $\varphi_{\phi}(b_1 b_2) = \varphi_{\phi}(b_1) \varphi_{\phi}(b_2)$ for all $b_1, b_2 \in B$.  By the above calculation, this holds \ifft $f(b_1, b_2)^{-1} f(\beta^{-1} b_1, \beta^{-1} b_2) = 1$ for all $b_1, b_2 \in B$.
\end{proof}

Next, we show that every automorphism of $G_f$ that fixes the sets $A$ and $\ell[B]$ induces an isometry of $f$.

\begin{lemma}
  \label{lem:aut-isom}
  Let $f : B \times B \ra A$ be a bilinear map and let $\phi \in \aut(G_f)$ such that $\phi[A] = A$ and $\phi[\ell[B]] = \ell[B]$.  Then $\beta = \restr{\phi}{B} : B \ra B$ is an isometry of $f$.
\end{lemma}

\begin{proof}
  Let $b \in B$.  Then $\phi \ell(b) \in \phi[\ell(b) A] = \ell(\beta b) A$.  Since this holds for all $b \in B$, we have $\phi \ell(\beta^{-1} b) \in \ell(b) A$.  Since $\phi[\ell[B]] = \ell[B]$, we see that in fact, $\phi \ell(\beta^{-1} b) = \ell(b)$.  Hence, $\varphi_{\phi} = 1$, which is a homomorphism.  \lemref{aut-hom} then implies that $\beta$ is an isometry of $f$.
\end{proof}

We are now ready to reduce \BI\spc to \GRIS.

\bigris*

\begin{proof}
  Let $f : B \times B \ra A$ be an instance of \BI\spc and construct the group $G_f$.  The order of this group is $\abs{A} \abs{B}$, which is polynomial in the size of our instance of \BI\spc by \defref{bi}.  We wish to compute the subgroup $\aut(G_f)_{A, \ell[B]}$ of $\aut(G_f)$ that maps $A$ to $A$ and $\ell[B]$ to $\ell[B]$.  This is a subgroup of $\sym(B) \times \sym(A)$ that acts on $G_f$ by $(\pi, \sigma)(\ell(b) a) = \ell(\pi b) \sigma a$.  Therefore, we can compute $\aut(G_f)_{A, \ell[B]}$ of $\aut(G_f)$ by solving a \GRIS\spc problem for the group $\sym(B) \times \sym(A)$.

  By \lemref{aut-isom}, for each $\phi \in \aut(G_f)_{A, \ell[B]}$, $\beta = \restr{\phi}{B}$ is an isometry of $f$.  Moreover, by \lemref{isom-aut}, for every isometry $\beta : B \ra B$ of $f$, there is an automorphism $\phi \in \aut(G_f)_{A, \ell[B]}$ such that $\beta = \restr{\phi}{B}$.  It follows that we can compute the isometry group of $f$ from $\aut(G_f)_{A, \ell[B]}$ in polynomial time.  Therefore, \BI\spc is many-one reducible to \GRIS\spc in polynomial time.
\end{proof}

% \section{Conclusion}
% \label{sec:conclusion}
% \input{conclusion}

\section*{Acknowledgements}
We thank Laci Babai for discussing his Holomorph idea with us.  FLG is supported by the Grant-in-Aid for Young Scientists~(A)~No.~16H05853 and the Grant-in-Aid for Scientific Research~(A)~No.~16H01705 of the Japan Society for the Promotion of Science, and the Grant-in-Aid for Scientific Research on Innovative Areas~No.~24106009 of the Ministry of Education, Culture, Sports, Science and Technology in Japan. DJR was funded by a Japan Society for the Promotion of Science Postdoctoral Fellowship No. PE15020.

\inlong{
  \appendix
  \section{Equivalence of $\CI$ and $\GIS$}
  \label{app:color-graph}
  We now give the proof of \thmref{cigis}.

\cigis*

\begin{proof}
  We start with the reduction from \GIS\spc to \CI\spc.  Let $X$ and $Y$ be graphs and suppose that we wish to find an isomorphism from $X$ to $Y$ in a subcoset $\sigma \Gamma$ the maps the vertices of $X$ to those of $Y$.  The reduction is immediate once we extend each $\sigma \pi \in \sigma \Gamma$ to map\footnote{For the purposes of the color isomorphism problem, we regard each $\sigma \pi$ as a permutation of $X \uplus Y$.} $X \times X$ to $Y \times Y$ and define colors according to the graphs $X$ and $Y$.

  The reduction from \CI\spc to \GIS\spc is slightly more complicated.  Let $\sigma \Gamma$ be a subcoset of permutations acting on a set $X$ and let $f_1 : X \ra [n]$ and $f_2 : X \ra [n]$ be as in \defref{color-iso}.  We define graphs $X_i$ for $i \in \{1, 2\}$ as follows.  The vertices of $X_i$ consist of the vertices of $X$ as well as certain gadgets that encode the colors.  For each distinct $f_i(x) \in f_i[X]$, we add a vertex labelled $f_i(x)$ and create a copy $K_{f_i(x)}$ of the complete graph on $f_i(x) + 2$ colors.  We add an edge from the vertex $f_i(x)$ to every vertex in $K_{f_i(x)}$.  Finally, we connect every vertex $x' \not= x \in X$ such that $f_i(x') = f_i(x)$ to the vertex $f_i(x)$.  Since complete subgraphs of size $3$ or larger appear only as the gadgets $K_{f_i(x)} \uplus \{f_i(x)\}$, it is easy to see that the graph isomorphisms from $X_1$ to $X_2$ correspond precisely to the color isomorphisms of the set $X$.  This completes the reduction.
\end{proof}

}

\inlong{\newpage}
\bibliographystyle{initials}
\bibliography{$HOME/LaTeX/computer-science-references,$HOME/LaTeX/math-references,$HOME/LaTeX/quantum-computing-references} %$

\end{document}